\documentclass[journal,onecolumn]{IEEEtran}

\usepackage{amssymb}
\usepackage{caption2}
\usepackage{amsmath}
\usepackage{multirow}
\usepackage{amsthm}

\newtheorem{theorem}{Theorem}[section]
\newtheorem{definition}{Definition}[section]
\newtheorem{lemma}[theorem]{Lemma}

\newtheorem{corollary}[theorem]{Corollary}

\newtheorem{remark}{Remark}[section]

\newcommand{\ma}{\mathcal}

\begin{document}

\title{New Bounds For Frameproof Codes}
\author{Chong Shangguan, Xin Wang, Gennian Ge and Ying Miao
\thanks{The research of G. Ge was supported by the National Natural Science Foundation of China under Grant No.~61171198 and Grant No.~11431003, and Zhejiang Provincial Natural Science Foundation of China under Grant No.~LZ13A010001.
The research of Y. Miao was supported by JSPS Grant-in-Aid for Scientific Research (C) under Grant No.~24540111.
}
\thanks{C. Shangguan is with the Department of Mathematics, Zhejiang University,
Hangzhou 310027, China (e-mail: 11235061@zju.edu.cn).}
\thanks{X. Wang is with the Department of Mathematics, Zhejiang University,
Hangzhou 310027, China (e-mail: 11235062@zju.edu.cn).}
\thanks{G. Ge is with the School of Mathematical Sciences, Capital Normal University,
Beijing 100048, China (e-mail: gnge@zju.edu.cn). He is also with Beijing Center for Mathematics and Information Interdisciplinary Sciences, Beijing, 100048, China.}
\thanks{Y. Miao is with the Faculty of Engineering, Information and Systems, University of Tsukuba,
Tsukuba, Ibaraki 305-8573, Japan (miao@sk.tsukuba.ac.jp).}
\thanks{Copyright (c) 2014 IEEE. Personal use of this material is permitted. However, permission to use this material for any other purposes must be obtained from the IEEE by sending a request to pubs-permissions@ieee.org.}
}
\maketitle

\begin{abstract}
Frameproof codes are used to fingerprint digital data. It can prevent copyrighted materials from unauthorized use. In this paper, we study upper and lower bounds for $w$-frameproof codes of length $N$ over an alphabet of size $q$. The upper bound is based on a combinatorial approach and the lower bound is based on a probabilistic
construction. Both bounds can improve previous results when $q$ is small compared to $w$, say $cq\leq w$ for some constant $c\leq q$. Furthermore, we pay special attention to binary frameproof codes. We show a binary $w$-frameproof code of length $N$ can not have more than $N$ codewords if $N<\binom{w+1}{2}$.
\end{abstract}

\begin{keywords}
Fingerprinting, frameproof code, combinatorial counting, deletion method
\end{keywords}

\section{Introduction}
\label{intro}

Frameproof codes were first introduced by Boneh and Shaw \cite{Firstpaper} in the context of digital fingerprinting to prevent copyrighted materials from unauthorized use. A set of unique labels, known as digital fingerprints, are inserted into digital data before the distributor wants to sell the copies to different customers. The fingerprints can be viewed as decoder boxes or access control for some copyrighted materials. They are generally denoted as codewords in $\ma{A}^N$, where $\ma{A}$ is an alphabet of size $q$ and $N$ is a positive integer.
As long as a coalition of some users want to produce a pirate decoder, they can share and compare their copies. A set of fingerprints is called to be a $w$-frameproof code if any coalition of at most $w$ users can not frame another user not in this coalition.

\subsection{Previous work}
Consider a code $\mathcal{C}\subseteq\mathcal{A}^N$.
Without loss of generality, we set $\ma{A}=\{1,\ldots,q\}$ for $q\geq3$ through out this paper. When considering the binary case, we usually set $\ma{A}=\{0,1\}$. We call this code an $(N,n,q)$ code if $|\ma{C}|=n$. Each codeword $\vec{c}\in\ma{C}$ can be represented as $\vec{c}=(c_1,\ldots,c_N)$, where $1\leq c_i\leq q$ for all $1\leq i\leq N$. For any subset of codewords $D\subseteq\ma{C}$, we denote $desc_i(D)=\{c_i:\vec{c}\in D\}$ for every $1\leq i\leq N$.
The set of $descendants$ of $D$ is defined as
$$desc(D)=\{\vec{c}\in\ma{A}^N:c_i\in desc_i(D),~1\leq i\leq N\}.$$
Let $\ma{C}$ be an $(N,n,q)$ code and let $w\geq2$ be an integer. $\ma{C}$ is called a $w$-frameproof code if we have $desc(D)\cap\ma{C}=D$ for all $D\subseteq\ma{C}$ and $|D|\leq w$. In the literature, there are a lot of papers about the properties and applications of frameproof codes, see for example \cite{BL03}, \cite{Firstpaper},
\cite{shijie}, \cite{chenmq}, \cite{BN}, \cite{ST01}, \cite{applications}, \cite{propertyandconstruction}, \cite{inviewedofshf}.
It is worth mentioning that frameproof codes were widely considered having no traceability for generic digital fingerprinting. However, in \cite{chenmq} the authors showed that frameproof codes have very good traceability for multimedia fingerprinting.
There are also many objects related to frameproof codes, such as identifiable parent property codes \cite{IPP1}, \cite{IPP2}, traceability codes \cite{TA1} and separating hash families \cite{SHF1}, \cite{SHF2}, \cite{ST08}.

The determination of upper and lower bounds of frameproof codes is crucial problems in this research area.
For upper bounds, when $q\geq w$ many strong bounds can be found in \cite{BL03}, \cite{ST01} and \cite{inviewedofshf}. Let $M_{w,N}(q)$ be the largest cardinality of an $(N,n,q)$ $w$-frameproof code and let $R_{w,N}=\lim_{q\rightarrow\infty} M_{w,N}(q)/q^{\lceil N/w\rceil}$. It has been determined by Blackburn \cite{BL03} that $\lim_{q\rightarrow\infty}\log_q M_{w,N}(q)=\lceil N/w\rceil$, $R_{w,N}=1$ when $N\equiv1\mod w$, $R_{w,N}=2$ when $w=2$ and $N$ is even.
Existing constructions are usually based on some finite fields with cardinality larger than
$N$, in other words, the alphabet size $q$ is larger than the code length $N$. Notice that this setting does not work if one wants to know about the code rate defined as $\alpha_{w,q}=\lim_{N\rightarrow\infty}\frac{1}{N}\log_q M_{w,N}(q)$,
where $q$ is some given positive integer. When $q<w$, much less is known about the upper bounds.
In this paper we will concentrate on this theme. For lower bounds, in \cite{ST08}, Stinson, Wei and
Chen presented a general result by probabilistic method, where frameproof codes were viewed as a special type of separating hash families.

These bounds are restated as follows.

\begin{theorem}(\cite{BL03}) \label{upperbound1}
Let $t\in\{1,\ldots,w\}$ be an integer such that $t\equiv N\mod w$. If there exists an $(N,n,q)$ $w$-frameproof code, then
$$n\leq(\frac{N}{N-(t-1)\lceil N/w\rceil})q^{\lceil N/w\rceil}+\ma{O}(q^{\lceil N/w\rceil-1}).$$
\end{theorem}

\begin{theorem} (\cite{ST08}) \label{generallowerbd} There exists an $(N,n,q)$ $w$-frameproof code provided that
$$n\leq (1-\frac{1}{w!})(\frac{q^w}{q^w-(q-1)^w})^{\frac{N}{w}}.$$
\end{theorem}

\begin{remark} In the original paper \cite{ST08}, this bound is stated in the form of separating hash family. Note that a separating hash family of type $(N,n,q,\{1,w\})$ is equivalent to an $(N,n,q)$ $w$-frameproof code.
\end{remark}

In a recent paper of Guo, Stinson and Trung \cite{BN}, the authors paid particular attention to binary frameproof codes with small code length.
They showed that if there exits an $(N,n,2)$ $w$-frameproof code $\ma{C}$ with $w\geq3$ and $w+1\leq N\leq3w$, then it always holds that
$n\leq N$. The equality holds if and only if the representation matrix of $\ma{C}$ in standard form is a permutation matrix of degree $N$.

\begin{theorem}(\cite{BN}) \label{binarycase} For all $w\geq 3$ and for all $w+1\leq N\leq 3w$, an
$(N,n,2)$ $w$-frameproof code exists only if $n\leq N$.
If $n=N$, then the representation matrix in standard form must be equivalent to a permutation matrix of degree $N$.
\end{theorem}

\subsection{Main results}
In this paper, we mainly investigate the bounds for frameproof codes. We focus on the situation
when $q$ is relatively small compared to $w$. There are few results on this case in previous papers. We study the behavior of the code rate
$\alpha_{w,q}=\lim_{N\rightarrow\infty}\frac{1}{N}\log_q M_{w,N}(q)$
as $N$ approximates infinity. Our upper and lower bounds are both better than previous results when $w\geq cq$, where $c$ is some constant. Given an alphabet size $q$ and if $w$ satisfies the above condition, the upper bound of $\alpha_{w,q}$ can be improved from $\ma{O}(1/w)$ in \cite{BL03} to $\ma{O}(\log w/w^2)$. For the lower bound, we show there always exits a code with
$\alpha_{w,q}=\Omega(1/w^2)$. For binary $w$-frameproof code of length $N$, we prove that if $N<\binom{w+1}{2}$, then
it always holds that $n\leq N$, which is a significant improvement of Theorem \ref{binarycase}.

This paper is organized as follows. Section 2 is about some definitions and lemmas. In Section 3, we prove our upper bound by a combinatorial counting argument. In Section 4, we present a tight bound for binary $w$-frameproof code with code length bounded by $\binom{w+1}{2}$. For the lower bound, we present our probabilistic construction in Section 5. Section 6 is about some final discussions and problems.

\section{Preliminaries}
We will first give some definitions and notations. In Section~\ref{intro} we already
presented a definition of frameproof code in terms of the descendant set, which is the most common one in the literature. But in this paper we prefer an alternative definition in terms of distance.

\begin{definition} \label{newdefinition}
For an $(N,n,q)$ code $\ma{C}$, the distance between any codeword $\vec{c}\in\ma{C}$ and any collection of codewords $D\subseteq\ma{C}$ is defined as follows:
$$d(\vec{c},D)=|\{i:1\leq i\leq N,~c_i\not\in desc_i(D)\}|.$$
$\mathcal{C}$ is called a $w$-frameproof code if $d(\vec{c},D)>0$ for all $|D|\leq w$ and $\vec{c}\not\in D$.
\end{definition}

It is easy to verify that this definition is equivalent to the original one. It is interesting that our ideas in this paper are indeed implied in this distance-based definition. Recall that a codeword $\vec{c}\in\ma{C}$ is denoted as $\vec{c}=(c_1,\ldots,c_N)$, where $1\leq c_i\leq q$ for all $1\leq i\leq N$. We can associate each codeword $\vec{c}$ with a set of two tuples $F_{\vec{c}}=\{(i,c_i):1\leq i\leq N\}$. Note that $F_{\vec{c}_1}=F_{\vec{c}_2}$ if and only if $\vec{c}_1=\vec{c}_2$ and it always holds that $|F_{\vec{c}}|=N$. We call $F_{\vec{c}}$ the set corresponding to $\vec{c}$. It is easy to show $|\cup_{\vec{c}\in\ma{C}}F_{\vec{c}}|\leq qN$. The following lemma establishes an immediate connection between the distance defined in Definition \ref{newdefinition} and  the cardinality of $F_{\vec{c}}$.

\begin{lemma} \label{distanceandset}
Suppose $D\subseteq\ma{C}$ is a set of codewords. For every $\vec{c}\in\ma{C}$, let $F_{\vec{c}}=\{(i,c_i):1\leq i\leq N\}$ be its corresponding set. Then we have
$$d(\vec{c},D)=|F_{\vec{c}}\setminus\cup_{\vec{d}\in D}F_{\vec{d}}|.$$
\end{lemma}

\begin{proof} Note that $d(\vec{c},D)=|\{i:1\leq i\leq N,~c_i\not\in desc_i(D)\}|$. The lemma follows from the simple fact that $c_i\in desc_i(D)$ if and only if $(i,c_i)\in F_{\vec{d}}$ for some $\vec{d}\in D$.
\end{proof}

\begin{remark}
One can verify that $\ma{C}$ is a $w$-frameproof code if and only if $|F_{\vec{c}}\setminus\cup_{\vec{d}\in D}F_{\vec{d}}|>0$ for all $D\subseteq\ma{C}$, $|D|\leq w$ and $\vec{c}\not\in D$ directly from Definition \ref{newdefinition}.
\end{remark}

For a set $S\subseteq\{1,\ldots,N\}$ and $|S|=t$, we call $T=\{(i,c_i): i\in S\}\subseteq F_{\vec{c}}$ a $t$-pattern of $\vec{c}$ (or equivalently, $F_{\vec{c}}$). $T$ is said to be an own $t$-pattern of $\vec{c}$
if $T\nsubseteq F_{\vec{d}}$ for all $\vec{d}\in\ma{C}\setminus\{\vec{c}\}$.
Denote $\mathcal{C}_t = \{\vec{c}\in\mathcal{C}: \vec{c}~\mbox{has~at~least~one~own}~t\textrm{-}{\mbox pattern}\}$
and $\mathcal{H}_t=\mathcal{C}-\mathcal{C}_t$. We have the following two easy lemmas.

\begin{lemma}\label{bdofct}
$|\mathcal{C}_t|\leq\binom{N}{t}q^t$.
\end{lemma}
\begin{proof} Just notice that every codeword $\vec{c}\in\mathcal{C}_t$ is uniquely identified by its arbitrary own $t$-pattern.
\end{proof}

\begin{lemma}\label{mainlemma}
Suppose $\vec{c}\in\mathcal{H}_t$. Then for any $1\leq j\leq w$ and any distinct $\vec{c}_1,\vec{c}_2,\ldots,\vec{c}_j\in\mathcal{C}\setminus\{\vec{c}\}$. We have
$$d(\vec{c},\{\vec{c}_1,\vec{c}_2,\ldots,\vec{c}_j\})=|F_{\vec{c}}\setminus\cup_{1\leq i\leq j}F_{\vec{c}_i}|\geq(w-j)t+1.$$
\end{lemma}

\begin{proof} Assume not. Then
$|F_{\vec{c}}\setminus\cup_{1\leq i\leq j}F_{\vec{c}_i}|\leq(w-j)t$ holds
for some $\vec{c}_1,\ldots,\vec{c}_j\in\ma{C}\setminus\{\vec{c}\}$. Note that $\vec{c}\in\ma{H}_t$, which means any $t$-pattern of $F_{\vec{c}}$ is not an own $t$-pattern, that is, it belongs to some $F_{\vec{d}}$ with $\vec{d}\in\ma{C}\setminus\{\vec{c}\}$.
Then there exist
$w-j$ codewords $\vec{c}_{j+1},\ldots,\vec{c}_{w}\in\ma{C}\setminus\{\vec{c}\}$ such that $F_{\vec{c}}\setminus\cup_{1\leq i\leq j}F_{\vec{c}_i}\subseteq\cup_{j+1\leq i\leq w}F_{\vec{c}_i}$. It implies
$|F_{\vec{c}}\setminus\cup_{1\leq i\leq w}F_{\vec{c}_i}|=0$, which contradicts the definition of $w$-frameproof code.
\end{proof}

\section{The upper bound}
In this section we present our upper bound. For an $(N,n,q)$
$w$-frameproof code $\ma{C}$, we study the behavior of the code rate defined as $\alpha_{w,q}=\lim_{N\rightarrow\infty}\frac{1}{N}\log_q|\ma{C}|$. Given an alphabet size $q$, in \cite{BL03} the author
proved that $\alpha_{w,q}=\ma{O}(1/w)$. Here we show $\alpha_{w,q}=\ma{O}(\log w/w^2)$. Our bound is always better if $w\geq cq$ for some constant $c$. The following theorem is based on previous lemmas with a particular observation of $|\cup_{\vec{c}\in\ma{C}}F_{\vec{c}}|\leq qN$.

\begin{theorem}\label{upperbound} Suppose~$\mathcal{C}\subseteq\mathcal{A}^N$ is an
$(N,n,q)$ $w$-frameproof code. Then we have
$$|\mathcal{C}| \leq q^{{\lceil N(q-1)/\binom{w}{2}}\rceil \log_{q}{eq\binom{w}{2}/(q-1)}}+w.$$
\end{theorem}

\begin{proof}
Suppose $|\ma{H}_t|\geq w+1$. Then for distinct $\vec{c}_1,\ldots,\vec{c}_{w+1}\in\mathcal{H}_t$ and the corresponding $F_{\vec{c}_i}$, by Lemma \ref{mainlemma} one has
$|\cup_{1\leq i\leq w+1}F_{\vec{c}_i}|=
|F_{\vec{c}_1}|+|F_{\vec{c}_2}\setminus F_{\vec{c}_1}|+\cdots+|F_{\vec{c}_{w+1}}\setminus(F_{\vec{c}_1}\cup\cdots\cup F_{\vec{c}_w})|
\geq N+w+\Sigma^{w}_{j=1}(w-j)t=N+w+t\binom{w}{2}$. The right hand side of above formula exceeds $qN$ for $t>\frac{N(q-1)-w}{\binom{w}{2}}$, which violates the fact $|\cup_{1\leq i\leq w+1}F_{\vec{c}_i}|\leq|\cup_{\vec{c}\in\ma{C}}F_{\vec{c}}|\leq qN$. This implies $|\mathcal{H}_t|\leq w$ for those
$t>\frac{N(q-1)-w}{\binom{w}{2}}$.
Let $t=\lceil\frac{N(q-1)}{\binom{w}{2}}\rceil$. Then we get
$$|\mathcal{C}|=|\mathcal{C}_t|+|\mathcal{H}_t|$$
$$\leq|\mathcal{C}_t|+w\leq\binom{N}{t}q^t+w$$
$$=\binom{N}{\lceil\frac{N(q-1)}{\binom{w}{2}}\rceil}q^{\lceil\frac{N(q-1)}{\binom{w}{2}}\rceil}+w$$
$$<(\frac{eqN}{\lceil\frac{N(q-1)}{\binom{w}{2}}\rceil})^{\lceil\frac{N(q-1)}{\binom{w}{2}}\rceil}+w$$
$$\leq q^{\lceil\frac{N(q-1)}{\binom{w}{2}}\rceil\log_{q}{eq\binom{w}{2}/(q-1)}}+w$$
from Lemma \ref{bdofct} using $\binom{N}{t}<(eN/t)^t$.
\end{proof}

\begin{remark}
Compared with Theorem \ref{upperbound1}, our bound is an improvement when
$\lceil\frac{ N(q-1)}{\binom{w}{2}}\rceil\log_{q}{eq\binom{w}{2}/(q-1)}<\lceil N/w\rceil$.
For positive real numbers
$x$ and $y>1$, it can be verified that $\lceil\frac{x}{y}\rceil<\lceil x\rceil\frac{1}{y-1}$ if $x$ is sufficiently large, for example,
$x>y^2$. Thus it can be verified that our code rate is better than Theorem \ref{upperbound1} when $w\geq 14q$ with $q\geq 14$. For $2\leq q\leq 14$,
in Table 1 we list the conditions when our bound is better.
\end{remark}

\begin{table}[ht]
\centering
\caption{Given $q$, the minimal $w$ that our bound is better}

\begin{tabular}{|l|l|l|l|l|}
  \hline
  % after \\: \hline or \cline{col1-col2} \cline{col3-col4} ...
  $q$ & 2 & 3 & 4 & 5 \\\hline
  min $w$ & 25 & 33 & 42 & 51 \\\hline
  $q$ & 6 & 7 & 8 & 9 \\\hline
  min $w$ & 51 & 60 & 68 & 77 \\\hline
  $q$ & 10 & 11 & 12 & 13 \\\hline
  min $w$ & 94 & 102 & 110 & 118 \\
  \hline
\end{tabular}
\end{table}
If the alphabet size $q$ is given, the previous known results often give an upper bound of code rate as $\alpha_{w,q}=\ma{O}(1/w)$. But our bound is$\alpha_{w,q}=\ma{O}(\log w/w^2)=\ma{O}(1/w^{2-\epsilon})$ where $\epsilon$ is some small quantity related to $w$. This difference is quite reasonable since in the literature there exits good constructions only when $q>N$, using some finite fields with size larger than $N$. If $q$ is relatively small, there exists no such good constructions and the explicit upper bound is still far from known. In Section~\ref{lower}, we will present a probabilistic construction with $\alpha_{w,q}=\Omega(1/w^2)$, which implies $c_1/w^2\leq\alpha_{w,q}\leq c_2/w^{2-\epsilon}$ for given $q$ and sufficiently large $w$, where $c_1$ and $c_2$ are positive real numbers.

\section{Tight bound on small code length}
In a recent paper of Guo, Stinson and Trung \cite{BN}, the authors considered the condition when a binary $w$-frameproof code of length $N$ can have more than $N$ codewords. They proved that if $\ma{C}$ is an $(N,n,2)$ $w$-frameproof code with $w+1\leq N\leq3w$ and $w\geq3$, then $n\leq N$. We will show $n\leq N$ still holds even when $N<\binom{w+1}{2}$.

We begin with some definitions which are originally from \cite{BN}. We can depict an $(N,n,q)$ code as an $N\times n$ matrix on $q$ symbols, where each column of the matrix corresponds to one of the codewords. This matrix is called the representation matrix of the code. Consider the representation matrix of any frameproof code. If we permute the entries in each row separately, i.e. a permutation of the elements $\{1,\ldots,q\}$, we get new frameproof codes which can be considered to be in the same equivalence class with the original one.
For binary code, we say it is in {\it standard form}
if every row of its representation matrix has at most $n/2$ entries of $1$.

Recall that $\ma{C}=\ma{C}_t\cup\ma{H}_t$. The following theorem is proved by considering the situation of $t=1$, in other words, the $1$-patterns.

\begin{theorem}\label{mybinarycase} Suppose $\ma{C}$ is an $(N,n,2)$ $w$-frameproof code with $w\geq 2$ and $N<\binom{w+1}{2}$.
Then it always holds that $n\leq N$. The equality holds if and only if the representation matrix of $\ma{C}$ in standard form is equivalent to a permutation matrix of order $N$.
\end{theorem}

In order to present our proof, we would like to give some lemmas first.

\begin{lemma} \label{1}Suppose $\ma{C}$ is an $(N,n,2)$ $w$-frameproof code with $N\leq w$. Then
it always holds that
$\ma{C}=\ma{C}_1$.
\end{lemma}
\begin{proof} The conclusion holds since if there is any $\vec{c}\in\ma{H}_1$, then for each
$1\leq i\leq N$ we can find some $\vec{d}_i\in\ma{C}\setminus\{\vec{c}\}$ such that $c_i=d_i$, implying $d({\vec{c},\{\vec{d}_1,...,\vec{d}_N\}})=0$.
\end{proof}

\begin{lemma}\label{2} Suppose $\ma{C}$ is an $(N,n,2)$ $w$-frameproof code with $\ma{C}=\ma{C}_1$. Then it always holds that $n\leq N$. The equality holds if and only if the representation matrix of $\ma{C}$ in standard form is equivalent to a permutation matrix of order $N$.
\end{lemma}
\begin{proof} If $\ma{C}=\ma{C}_1$, then for every codeword $\vec{c}\in\ma{C}$, there exists some $i$, $1\leq i\leq N$, such that $c_i\neq d_i$ for all $\vec{d} \in \ma{C} \setminus \{\vec{c}\}$. Consider the representation matrix of $\ma{C}$. Since each $c_i\in\{0,1\}$, then we can set this special $c_i$ to be $1$ and other entries in row $i$ to be $0$. Thus $n\leq N$ since there is at least one such row for each $\vec{c}\in\ma{C}_1$. The equality holds if and only if there is exactly one such row for each $\vec{c}\in\ma{C}_1$. The resulting representation matrix is a permutation matrix of order $N$.
\end{proof}

Now we can prove Theorem \ref{mybinarycase}.
\begin{proof}
The case $N\leq w$ follows from Lemma \ref{1} and Lemma \ref{2}.
If $N\geq w+1$, then we apply induction on $N$.  The proof is divided into three parts.
\begin{enumerate}
  \item If $\ma{C}_1=\emptyset$, then $\ma{C}=\ma{H}_1$.
  Assume that $|\ma{H}_1|\geq w+1$.
  Then we can choose $w+1$ distinct codewords
  $\vec{c}_1, \ldots, \vec{c}_{w+1}\in\ma{H}_1$.
  By setting $t=1$ in Lemma \ref{mainlemma}, we have
      $2N\geq |\cup_{1\leq i\leq w+1}F_{\vec{c}_i}|=
      |F_{\vec{c}_1}|+|F_{\vec{c}_2}\setminus F_{\vec{c}_1}|+\cdots+|F_{\vec{c}_{w+1}}\setminus(F_{\vec{c}_1}\cup\cdots\cup F_{\vec{c}_w})|
      \geq N+w+\Sigma^{w}_{j=1}(w-j)=N+w+\binom{w}{2}=N+\binom{w+1}{2}$,
      which violates the assumption $N<\binom{w+1}{2}$.
      Therefore, $|\ma{H}_1|\leq w$. Then the conclusion follows from $N\geq w+1$.
  \item If $\ma{C}_1\neq\emptyset$ and $\vec{c}\in\ma{C}_1$. Then there exists some $i$, $1\leq i\leq N$, such that $c_i\neq d_i$ for all $\vec{d}\in \ma{C} \setminus \{\vec{c}\}$.
  Consider the standard representation matrix of $\ma{C}$. Since $c_i\in\{0,1\}$, then there is exactly one
  $1$ in the $i$-th row of this matrix, which implies $c_i = 1$.
  By permuting the row entries, we can set the other entries of the column indexed by $\vec{c}$ to be zero. Let $i=1$ for simplicity. Then by above discussion, we actually get an equivalent code that has a codeword $\vec{c}=(1,0,\ldots,0)$. The representation matrix is of the following form,
      $$1~~~0~~~0~~\cdots\cdots0~~~0$$
      $$0~~*~~*~~\cdots\cdots*~~*$$
      $$\cdots\cdots$$
      $$0~~*~~*~~\cdots\cdots*~~*$$
      where * denotes a symbol in $\{0,1\}$. By deleting the first row and first column, we can get a new matrix which is the representation matrix of an $(N-1,n-1,2)$ $w$-frameproof code. Let us denote this code by $\ma{C}^{(1)}$.
   \item Repeat the procedures in step 2 for $\ma{C}^{(1)}$. If we always have $\ma{C}^{(i)}_1\neq\emptyset$ as $i$ grows, then we will end up with a matrix where any row contains at most one $1$, which implies that $n \leq N$, and the equality holds if and only if this matrix is the standard representation matrix of the original code $\ma{C}$. Otherwise, there is some $\ma{C}^{(i)}$ such that the procedure ends up with some $\ma{C}_1^{(i)} = \emptyset$, that is, $\ma{C}^{(i)}=\ma{H}^{(i)}_1$. By Lemma \ref{1} the code length of $\ma{C}^{(i)}$ is at least $w+1$. Note that $\ma{C}^{(i)}$ is an $(N-i,n-i,2)$ $w$-frameproof code, then the conclusion follows from the arguments in step 1.
\end{enumerate}
\end{proof}

Denote $N(w)$ the minimal $N$ such that there exits an $(N,n,2)$ $w$-frameproof code with $n>N$. It was proved that $N(w)\geq w+1$ in \cite{BL03}. In \cite{BN}, the authors actually verified $N(2)=3$ and $N(w)\geq 3w+1$ for $w\geq 3$. In this paper it is improved to $N(w)\geq\binom{w+1}{2}$ for all $w\geq2$. In the following we will present an example that gives rise to $N(w)\leq w^2+o(w^2)$.

\begin{definition} An affine plane is an incidence system of points and lines such that
\begin{itemize}
  \item (AP1) For any two distinct points, there is exactly one line through both points.
  \item (AP2) Given any line $l$ and any point $P$ not on $l$, there is exactly one line through $P$ that does not meet $l$.
  \item (AP3) There exist four points such that no three are collinear.
\end{itemize}
\end{definition}

For a detailed introduction of affine plane, the readers are referred to \cite{affineplane}. In an affine plane, any two lines have the same number of points, finite or infinite. The order of an affine plane is the number of points on any given line of the plane. If $\ma{P}$ is an affine plane of finite order $r$, then it is proved that every point of $\ma{P}$ lies on exactly $r+1$ lines and $\ma{P}$ has exactly $r^2$ points and $r^2+r$ lines. Let $M$ be the incidence matrix of $\ma{P}$. Then $M$ is an $r^2\times (r^2+r)$ matrix whose rows are indexed by points of $\ma{P}$ and columns are indexed by lines of $\ma{P}$. The entry $M(P,l)=1$ if point $P$ is on line $l$ and $M(P,l)=0$ otherwise. By property (AP1), any two columns of $M$ can agree with at most one 1 in their coordinates. Consider the binary frameproof code with representation matrix $M$. It is easy to show this code is an $(r^2,r^2+r,2)$ $(r-1)$-frameproof code. For every prime power, it is known that there exists an affine plane of this order. Let $q$ be the smallest prime power no less than $w+1$. Then the existence of an affine plane of order $q$ gives rise to a $(q^2,q^2+q,2)$ $w$-frameproof code. Together with Theorem \ref{mybinarycase} we have the following corollary

\begin{corollary} Let $N(w)$ be the minimal $N$ such that there exists an $(N,n,2)$ $w$-frameproof code with $n>N$. Then
$$(\frac{1}{2}+o(1))w^2\leq N(w)\leq(1+o(1))w^2.$$
\end{corollary}

\section{The lower bound}
\label{lower}

In this section, we use the standard probabilistic method to give an existence result for frameproof codes. The technique we employ is commonly termed the deletion method.

\begin{theorem}\label{low}
If $q\leq w+1$, then there exists an
$(N,n,q)$
$w$-frameproof code with cardinality at least
  $$\frac{1}{2^{\frac{w+1}{w}}}(\frac{1}{1-(1-\frac{q-1}{w+1})(\frac{q-1}{w+1})^w-\frac{q-1}{w+1}(\frac{w}{w+1})^w})^{\frac{N}{w}}.$$
\end{theorem}

\begin{proof}
Assume
$\ma{A} = \{1, \ldots, q\}$
is an alphabet of size $q$. Then choose $2M$ vectors in $\mathcal{A}^N$ independently with repetitions.
Each vector $\vec{c}=(c_1,\ldots,c_N)$ is generated as follows. For every $1\leq i\leq N$, we set $Pr[c_i=q]=\lambda$ and $Pr[c_i=j]=\mu$ for each $1\leq j\leq q-1$.
Obviously
it holds that $\lambda+(q-1)\mu=1$.
The values $M$, $\lambda$ and $\mu$ will be determined later.
Denote the obtained random family by $\mathcal{C}_0$. For some $\vec{c}\in\mathcal{C}_0$ and arbitrary
$D\subseteq\mathcal{C}_0 \setminus \{\vec{c}\}$
with $|D|=w$, we compute the probability that $\vec{c}$ and $D$ violate the condition of $w$-frameproof code, i.e, $Pr[d(\vec{c},D)=0]$. We have
$$Pr[d(\vec{c},D)=0]=\prod_{1=1}^{N}Pr[c_i\in desc_i(D)],$$
where $$Pr[c_i\in desc_i(D)]=\lambda(1-(1-\lambda)^w)+(q-1)\mu(1-(1-\mu)^w)$$
holds for every $1\leq i\leq N$. Denote $Pr[c_i\in desc_i(D)]:=P(\lambda,\mu)$ for convenience. Then
$Pr[d(\vec{c},D)=0]=(P(\lambda,\mu))^N$.
Thus for
      $M<2^{-\frac{w+1}{w}}P(\lambda,\mu)^{-\frac{N}{w}}$,
      the expected number of pairs $\vec{c}$ and $D\subseteq\mathcal{C}_0 \setminus \{\vec{c}\}$ violating the property of $w$-frameproof code is bounded by
      $$\leq\binom{2M}{1}\binom{2M-1}{w}Pr[d(\vec{c},D)=0] $$
      $$ < (2M)^{w+1}(P(\lambda,\mu))^N < M,$$
      that is, there exists a code $\mathcal{C}_0$ of cardinality $|\mathcal{C}_0| = 2M$ with at most $M$ pairs $\vec{c}$ and $D\subseteq\mathcal{C}_0 \setminus \{\vec{c}\}$, $|D|=w$, violating the property of $w$-frameproof code.

Fix such a code $\mathcal{C}_0$ and for each of the pairs $(\vec{c}, D)$ delete the inadmissible vector $\vec{c}$. Denote the resulting code by $\mathcal{C}$. Then the expected number of the remaining vectors in $\mathcal{C}$ is greater than $2M-M=M$, and the vectors in $\mathcal{C}$ satisfy the condition of $w$-frameproof code.
It concludes that we have shown the existence of a $w$-frameproof code $\mathcal{C}$ with at least $M$ codewords.
Then the theorem follows by setting $\lambda=1-\frac{q-1}{w+1}$ and $\mu=\frac{1}{w+1}$.
\end{proof}

\begin{remark}
Assume $p_j$ is the probability that $Pr[c_i=j]=p_j$ for $1\leq 1\leq N$ and each $j\in\{1,\ldots,q\}$. The core of our construction is in fact the optimization
$$\max~~\sum_{j=1}^q p_j(1-p_j)^w$$
$$s.t.~~~~\sum_{j=1}^q p_j=1.$$
In \cite{ST08}, they just choose $p_1=p_2=\cdots=p_q=\frac{1}{q}$. Our choice is better when $q\leq w+1$.
\end{remark}

\begin{remark}
Compared with the bound in Theorem \ref{generallowerbd}, it is easy to see our bound on code rate is an improvement if
$(1-\frac{q-1}{w+1})(\frac{q-1}{w+1})^w+\frac{q-1}{w+1}(\frac{w}{w+1})^w>(\frac{q-1}{q})^w$. One can verify that the above inequality holds when $q\leq\frac{w}{2}+1$ if $w\geq 8$ (see Appendix). For small $w$ and $q$, in Table 2 we list the conditions when the inequality is satisfied.
\end{remark}

\begin{table}[ht]
\centering
\caption{Given $q$, the minimal $w$ that our bound is better}

\begin{tabular}{|l|l|l|l|}
\hline
  $q$ & 2 & 3 & 4  \\\hline
  min $w$ & 5 & 7 & 8  \\\hline
  $q$ & 5 & 40 & 41  \\\hline
  min $w$ & 8 & 49 & 50  \\\hline
\end{tabular}
\end{table}

Last two values indicate that the ratio of $q/w$ may approximate 1 as $q$ grows. We choose $q/w=1/2$ just for the sake of easy computing. By Taylor expansion of the function $\ln(1+x)$ one can show the rate of our lower bound is
$\Omega(1/w^2)$.
There is still a gap against the upper bound $\ma{O}(\log w/w^2)$ determined in Theorem \ref{upperbound}. To see how better our bound is than \cite{ST08} when $q$ is fixed and $w$ grows large, we just present both bounds for $q=2$, which are $\Omega(1/w2^w)$ and $\Omega(1/w^2)$ respectively.

\section{Conclusions}

In this paper we investigate the bounds of frameproof codes. To prove the upper bound we use a pure combinatorial approach. For the lower bound, we use a probabilistic method. Our main ideas are from the distance implied in the structure of frameproof codes. Compared with previous bounds, it is easy to find out that the cardinality of frameproof codes differs widely according to the numerical relationship of frameproof property $w$, alphabet size $q$ and code length $N$. For large $q$, the authors of \cite{BL03} gave several good constructions that fit in the upper bound. But it is not the case for small
$q$, where the gap between the lower and the upper bound is still huge. It is nice if one can construct frameproof codes on small alphabet with good code rate.
For binary $w$-frameproof code of length $N$, people pay particular attention to the extremal situation that one can find a code having at least $N+1$ codewords. In this paper we show the necessary condition is $N\geq\binom{w+1}{2}$. And it is showed that the sufficient condition is $N=w^2+o(w^2)$. Any improvement of this result may be interesting.

\begin{appendix}
We want to show $(1-\frac{q-1}{w+1})(\frac{q-1}{w+1})^w+\frac{q-1}{w+1}(\frac{w}{w+1})^w>(\frac{q-1}{q})^w$ when $q\leq w/2+1$ and $w\geq 8$.
It suffices to show $\frac{q-1}{w+1}(\frac{w}{w+1})^w>(\frac{q-1}{q})^w,$
which is equivalent to $\frac{w^w}{(w+1)^{w+1}}>\frac{(q-1)^{w-1}}{q^w}.$
Note that $f(x)=\frac{(x-1)^{w-1}}{x^w}$ is monotonically increasing when $1<x<w$.
Therefore,
$\frac{(q-1)^{w-1}}{q^w} < \frac{(w/2)^{w-1}}{(w/2+1)^w}.$
It suffices to show $\frac{w^w}{(w+1)^{w+1}}>\frac{2w^{w-1}}{(w+2)^w}$,
which is equivalent to $(1+\frac{1}{w+1})^{w+1}>\frac{2(w+2)}{w}$.
Note that the left side of the last inequality is monotonically increasing and right side is monotonically decreasing.
By direct computation one can show the inequality holds if $w\geq8$.
\end{appendix}
\bibliographystyle{IEEEtranS}
\bibliography{Reference}

% Generated by IEEEtranS.bst, version: 1.13 (2008/09/30)
\begin{thebibliography}{10}
\providecommand{\url}[1]{#1}
\csname url@samestyle\endcsname
\providecommand{\newblock}{\relax}
\providecommand{\bibinfo}[2]{#2}
\providecommand{\BIBentrySTDinterwordspacing}{\spaceskip=0pt\relax}
\providecommand{\BIBentryALTinterwordstretchfactor}{4}
\providecommand{\BIBentryALTinterwordspacing}{\spaceskip=\fontdimen2\font plus
\BIBentryALTinterwordstretchfactor\fontdimen3\font minus
  \fontdimen4\font\relax}
\providecommand{\BIBforeignlanguage}[2]{{%
\expandafter\ifx\csname l@#1\endcsname\relax
\typeout{** WARNING: IEEEtranS.bst: No hyphenation pattern has been}%
\typeout{** loaded for the language `#1'. Using the pattern for}%
\typeout{** the default language instead.}%
\else
\language=\csname l@#1\endcsname
\fi
#2}}
\providecommand{\BIBdecl}{\relax}
\BIBdecl

\bibitem{IPP1}
\BIBentryALTinterwordspacing
N.~Alon and U.~Stav, ``New bounds on parent-identifying codes: the case of
  multiple parents,'' \emph{Combin. Probab. Comput.}, vol.~13, no.~6, pp.
  795--807, 2004. [Online]. Available:
  \url{http://dx.doi.org/10.1017/S0963548304006388}
\BIBentrySTDinterwordspacing

\bibitem{SHF1}
\BIBentryALTinterwordspacing
M.~Bazrafshan and T.~Trung, ``Bounds for separating hash families,'' \emph{J.
  Combin. Theory Ser. A}, vol. 118, no.~3, pp. 1129--1135, 2011. [Online].
  Available: \url{http://dx.doi.org/10.1016/j.jcta.2010.11.006}
\BIBentrySTDinterwordspacing

\bibitem{BL03}
\BIBentryALTinterwordspacing
S.~R. Blackburn, ``Frameproof codes,'' \emph{SIAM J. Discrete Math.}, vol.~16,
  no.~3, pp. 499--510 (electronic), 2003. [Online]. Available:
  \url{http://dx.doi.org/10.1137/S0895480101384633}
\BIBentrySTDinterwordspacing

\bibitem{IPP2}
\BIBentryALTinterwordspacing
------, ``An upper bound on the size of a code with the {$k$}-identifiable
  parent property,'' \emph{J. Combin. Theory Ser. A}, vol. 102, no.~1, pp.
  179--185, 2003. [Online]. Available:
  \url{http://dx.doi.org/10.1016/S0097-3165(03)00030-X}
\BIBentrySTDinterwordspacing

\bibitem{TA1}
\BIBentryALTinterwordspacing
S.~R. Blackburn, T.~Etzion, and S.~Ng, ``Traceability codes,'' \emph{J. Combin.
  Theory Ser. A}, vol. 117, no.~8, pp. 1049--1057, 2010. [Online]. Available:
  \url{http://dx.doi.org/10.1016/j.jcta.2010.02.009}
\BIBentrySTDinterwordspacing

\bibitem{SHF2}
\BIBentryALTinterwordspacing
S.~R. Blackburn, T.~Etzion, D.~R. Stinson, and G.~M. Zaverucha, ``A bound on
  the size of separating hash families,'' \emph{J. Combin. Theory Ser. A}, vol.
  115, no.~7, pp. 1246--1256, 2008. [Online]. Available:
  \url{http://dx.doi.org/10.1016/j.jcta.2008.01.009}
\BIBentrySTDinterwordspacing

\bibitem{Firstpaper}
\BIBentryALTinterwordspacing
D.~Boneh and J.~Shaw, ``Collusion-secure fingerprinting for digital data,''
  \emph{IEEE Trans. Inform. Theory}, vol.~44, no.~5, pp. 1897--1905, 1998.
  [Online]. Available: \url{http://dx.doi.org/10.1109/18.705568}
\BIBentrySTDinterwordspacing

\bibitem{affineplane}
F.~Buekenhout, Ed., \emph{Handbook of incidence geometry}.\hskip 1em plus 0.5em
  minus 0.4em\relax North-Holland, Amsterdam, 1995, buildings and foundations.

\bibitem{shijie}
\BIBentryALTinterwordspacing
Y.~M. Chee and X.~Zhang, ``Improved constructions of frameproof codes,''
  \emph{IEEE Trans. Inform. Theory}, vol.~58, no.~8, pp. 5449--5453, 2012.
  [Online]. Available: \url{http://dx.doi.org/10.1109/TIT.2012.2197812}
\BIBentrySTDinterwordspacing

\bibitem{chenmq}
M.~Cheng and Y.~Miao, ``On anti-collusion codes and detection algorithms for
  multimedia fingerprinting,'' \emph{Information Theory, IEEE Transactions on},
  vol.~57, no.~7, pp. 4843--4851, July 2011.

\bibitem{BN}
C.~Guo, D.~R. Stinson, and T.~Trung, ``{On tight bounds for binary frameproof
  codes},'' \emph{http://arxiv.org/abs/1406.6920}, 2014.

\bibitem{ST01}
\BIBentryALTinterwordspacing
J.~N. Staddon, D.~R. Stinson, and R.~Wei, ``Combinatorial properties of
  frameproof and traceability codes,'' \emph{IEEE Trans. Inform. Theory},
  vol.~47, no.~3, pp. 1042--1049, 2001. [Online]. Available:
  \url{http://dx.doi.org/10.1109/18.915661}
\BIBentrySTDinterwordspacing

\bibitem{applications}
\BIBentryALTinterwordspacing
D.~R. Stinson, T.~Trung, and R.~Wei, ``Secure frameproof codes, key
  distribution patterns, group testing algorithms and related structures,''
  \emph{J. Statist. Plann. Inference}, vol.~86, no.~2, pp. 595--617, 2000,
  special issue in honor of Professor Ralph Stanton. [Online]. Available:
  \url{http://dx.doi.org/10.1016/S0378-3758(99)00131-7}
\BIBentrySTDinterwordspacing

\bibitem{propertyandconstruction}
\BIBentryALTinterwordspacing
D.~R. Stinson and R.~Wei, ``Combinatorial properties and constructions of
  traceability schemes and frameproof codes,'' \emph{SIAM J. Discrete Math.},
  vol.~11, no.~1, pp. 41--53 (electronic), 1998. [Online]. Available:
  \url{http://dx.doi.org/10.1137/S0895480196304246}
\BIBentrySTDinterwordspacing

\bibitem{ST08}
\BIBentryALTinterwordspacing
D.~R. Stinson, R.~Wei, and K.~Chen, ``On generalized separating hash
  families,'' \emph{J. Combin. Theory Ser. A}, vol. 115, no.~1, pp. 105--120,
  2008. [Online]. Available: \url{http://dx.doi.org/10.1016/j.jcta.2007.04.005}
\BIBentrySTDinterwordspacing

\bibitem{inviewedofshf}
\BIBentryALTinterwordspacing
T.~Trung, ``A tight bound for frameproof codes viewed in terms of separating
  hash families,'' \emph{Des. Codes Cryptogr.}, vol.~72, no.~3, pp. 713--718,
  2014. [Online]. Available: \url{http://dx.doi.org/10.1007/s10623-013-9800-0}
\BIBentrySTDinterwordspacing

\end{thebibliography}

\end{document}